\documentclass[runningheads]{llncs}
\usepackage[T1]{fontenc}
\usepackage{graphicx}
\usepackage{amsmath}
\usepackage{amssymb}
\usepackage{algorithm}
\usepackage{algpseudocode}
\usepackage{mathtools}
\usepackage{subcaption}
\usepackage{orcidlink}
\usepackage[misc,geometry]{ifsym}
\algrenewcommand\algorithmicrequire{\textbf{Input:}}
\algrenewcommand\algorithmicensure{\textbf{Output:}}
\newcommand{\R}{\mathbb{R}}
\renewcommand{\O}{\mathrm{O}}
\newcommand{\V}{\mathcal{V}}
\newcommand{\K}{\mathcal{K}}
\renewcommand{\L}{\mathcal{L}}
\newcommand{\diag}[1]{\mathrm{diag}\left(#1\right)}
\newcommand{\tr}[1]{\mathrm{tr}\left(#1\right)}
\newcommand{\lrp}[1]{\left(#1\right)}

\newcommand{\lrs}[1]{\left\{#1\right\}}
\newcommand{\Sym}{\mathrm{Sym}}
\newcommand{\Fl}{\mathrm{Fl}}

\newcommand{\C}{\mathrm{C}}
\newcommand{\T}{^\top}
\DeclareMathOperator*{\argmin}{argmin}
\newcommand{\norm}[1]{\left\lVert#1\right\rVert}
\begin{document}
\title{Eigengap Sparsity for Covariance Parsimony}  
%
\author{Tom Szwagier$^{(\text{\Letter})}$\orcidlink{0000-0002-2903-551X} \and
Guillaume Olikier\orcidlink{0000-0002-2767-0480} \and
Xavier Pennec\orcidlink{0000-0002-6617-7664}}
\authorrunning{T. Szwagier et al.}
%
\institute{Université Côte d'Azur, and Inria, Epione Project Team, Sophia-Antipolis, France
\email{\{tom.szwagier,guillaume.olikier,xavier.pennec\}@inria.fr}}
\maketitle              
\begin{abstract}
Covariance estimation is a central problem in statistics. An important issue is that there are rarely enough samples $n$ to accurately estimate the $p (p+1) / 2$ coefficients in dimension $p$. Parsimonious covariance models are therefore preferred, but the discrete nature of model selection makes inference computationally challenging. In this paper, we propose a relaxation of covariance parsimony termed ``eigengap sparsity'' and motivated by the good accuracy--parsimony tradeoffs of eigenvalue-equalization in covariance matrices. This penalty can be included in a penalized-likelihood framework that we propose to solve with a projected gradient descent on a monotone cone. The algorithm turns out to resemble an isotonic regression of mutually-attracted sample eigenvalues, drawing an interesting link between covariance parsimony and shrinkage.
\keywords{Covariance estimation \and Parsimony \and Eigengaps \and Flag manifolds \and Monotone cone \and Isotonic regression.}
\end{abstract}
\section{Introduction}
The \textit{principle of parsimony}, also known as Occam's razor (``The simplest explanation is usually the best one.''), is central in statistics. It implies that for two competing statistical models with similar likelihoods, one should choose the model with fewer parameters. It can be quantitatively achieved via model selection criteria such as the Bayesian Information Criterion (BIC)~\cite{schwarz_estimating_1978}, 
or the Akaike Information Criterion (AIC)~\cite{akaike_new_1974}.
Since the seminal LASSO paper~\cite{tibshirani_regression_1996}, it has been common to investigate parsimonious model estimation as a regularized optimization problem with a \textit{sparsity-inducing penalty}~\cite{bach_optimization_2012}.
Sparsity involves $\ell^0$-norms and leads to hard combinatorial problems, which is why it is classically \textit{relaxed} with $\ell^1$-norms, enjoying convergence and sparsity guarantees~\cite{bach_optimization_2012}.

The principle of parsimony finds a particular interest in the fundamental problem of \textit{covariance estimation}~\cite{pourahmadi_covariance_2011}. Indeed, in dimension $p$, covariance matrices have $p (p+1) / 2$ independent parameters, which makes them often \textit{overparameterized} with respect to the number $n$ of available samples.
The \textit{sample covariance matrix}---which is the maximum likelihood estimate of the population covariance matrix---then suffers from many \textit{sampling errors}. 
A typical high-dimensional phenomenon is the Marčenko--Pastur law~\cite{marcenko_distribution_1967}, which describes an important dispersion of the sample eigenvalues around their true value. This \textit{eigenvalue-spreading} phenomenon~\cite{johnstone_pca_2018} is illustrated in the seminal paper of Ledoit and Wolf~\cite[Fig.~4]{ledoit_well-conditioned_2004} and motivates the idea of \textit{shrinking} eigenvalues towards a mean value. To that extent, many regularized covariance estimation methods have emerged, notably involving shrinkage~\cite{stein_estimation_1975,ledoit_well-conditioned_2004,donoho_optimal_2018,ledoit_quadratic_2022}, sparsity~\cite{huang_covariance_2006,friedman_sparse_2008}, Bayesian methods~\cite{yang_estimation_1994,minka_automatic_2000} and random matrix theory~\cite{karoui_spectrum_2008}.

In this paper, we propose to tackle the problem of  covariance estimation with a \textit{penalized-likelihood} optimization on \textit{positive-definite} matrices: 
\begin{equation}\label{eq:problem}
    \hat \Sigma \in \argmin_{\Sigma \in \Sym_+(p)} \, \lrp{-2 \ln\L(\Sigma) + \alpha \operatorname{dim}(\Sigma)},
\end{equation}
where $\L\colon\Sym_+(p)\to\R$ is a likelihood function (e.g. Gaussian~\cite{tipping_probabilistic_1999,szwagier_curse_2024}, elliptical~\cite{chen_robust_2011} or Gaussian mixture~\cite{tipping_mixtures_1999,bouveyron_high-dimensional_2007,szwagier_parsimonious_2025}), $\alpha \in \R_{\geq 0}$ is a constant (like $~{\alpha_n = \ln n}$ for the BIC, $\alpha = 2$ for the AIC or any tuning hyperparameter) and $\operatorname{dim}\colon\Sym_+(p)\to\mathbb{Z}_{>0}$ is a penalty on the number of covariance parameters, i.e. the dimension of the submanifold of $\Sym_+(p)$ to which $\Sigma$ belongs, as properly defined later.
There are many ways to partition the space of covariance matrices into submanifolds (e.g. via coefficient-wise sparsity~\cite{friedman_sparse_2008} or rank constraints~\cite{bach_optimization_2012}). The one that we defend in this paper is the stratification of covariance matrices by the \textit{multiplicities of the eigenvalues}~\cite{groisser_geometric_2017}.
This original notion of parsimony is strongly motivated by the recent principal subspace analysis methodology~\cite{szwagier_curse_2024} (equalizing close eigenvalues yields a \textit{quadratic} decrease in the number of parameters without much decreasing the likelihood) which is a natural extension of the seminal probabilistic principal component analysis (PPCA) of Tipping and Bishop~\cite{tipping_probabilistic_1999}.

This paper is organized as follows. 
In section~\ref{sec:geometry}, we study the aforementioned stratification of the optimization space $\Sym_+(p)$, which enables to formally define the penalty $\operatorname{dim}$ in terms of eigenvalue multiplicities. We deduce that $\operatorname{dim}$ is constant on fixed-multiplicities submanifolds of $\Sym_+(p)$, which drastically simplifies the optimization problem. The bad news, however, is that solving the latter would require solving $2^{p-1}$ (unpenalized) optimization problems, which is computationally prohibitive, even in moderate dimensions. To that extent, we propose in section~\ref{sec:method} a relaxation of the penalty. 
This $\ell^1$-relaxation is based on the intriguing result that the number of covariance parameters can be \textit{exactly} rewritten as an $\ell^0$-norm of eigenvalue gaps. We will subsequently refer to our parsimony-inducing penalty as \textit{eigengap sparsity}.
We study the relaxed problem in the Gaussian setting in section~\ref{sec:Gaussian} and propose a projected gradient descent algorithm on the monotone cone of eigenvalues to address it. The algorithm draws an interesting link between covariance parsimony and shrinkage methods~\cite{ledoit_quadratic_2022}, and it shows promising results in synthetic experiments in section~\ref{sec:experiments}.

\section{Stratification of $\Sym_+(p)$ by Eigenvalue Multiplicities}\label{sec:geometry}
Let $\Sigma \in \Sym_+(p)$ eigendecompose as $\Sigma = \sum_{j=1}^p \lambda_j v_j {v_j}\T$ with $~{(\lambda_1, \dots, \lambda_p)\in\R^p}$ the {eigenvalues} (ordered-decreasing and positive) and $(v_1, \dots, v_p)\in\O(p)$ some associated {eigenvectors} ($\O(p)$ denotes the orthogonal group). If all the eigenvalues are \textit{simple}, then we need $p$ parameters to describe the eigenvalues and $\operatorname{dim}(\O(p)) = p (p-1) / 2$ parameters to describe the eigenvectors, resulting in a total of $\operatorname{dim}(\Sym_+(p)) = p (p+1) / 2$ covariance parameters to estimate.
Otherwise, some eigenvalues are \textit{multiple}, therefore we need fewer parameters not only for these eigenvalues, but also for the associated eigenspaces---since we quotient the rotational invariance of the eigenvectors within the eigenspace they span. 

More precisely, let $\C(p) \coloneqq \{\gamma\in\mathbb{Z}_{>0}^d\colon \sum_{k=1}^d \gamma_k = p, \, d\in[\![1, p]\!]\}$ be the set of \textit{compositions} of $p$~\cite{bergeron_standard_1995}.
Let $\gamma \coloneqq (\gamma_1, \dots, \gamma_d) \in \C(p)$, $q_k \coloneqq \sum_{l=1}^{k} \gamma_l$, and
\begin{equation}
    \K_+(\gamma) \coloneqq \{\lambda\in\R^p\colon \lambda_1=\dots=\lambda_{q_1} > \lambda_{q_1+1}=\dots=\lambda_{q_2} > \dots > 0\}
\end{equation}
be the set of \textit{ordered-decreasing} \textit{piecewise-constant} $p$-tuples of positive reals, whose pieces have respective sizes $\gamma_1, \gamma_2, \dots, \gamma_d$.
By eigenvalue decomposition,
\begin{equation}
    \Sym_+(\gamma) \coloneqq \lrs{V \diag{\lambda} V\T\colon \lambda \in \K_+(\gamma), V\in\O(p)}
\end{equation}
is the set of symmetric positive-definite matrices with eigenvalue multiplicities $\gamma_1, \dots, \gamma_d$. It is a \textit{submanifold} of $\Sym_+(p)$ of \textit{dimension} $d + (p^2 - \sum_{k=1}^d \gamma_k^2) / 2$. 

\begin{remark}
The theory behind this result is deep and has been worked out in~\cite{groisser_geometric_2017}. Let us still give a concise justification.
Let $\Fl(\gamma)$ be the \textit{flag manifold} of \textit{type} $\gamma$, i.e. the space of mutually-orthogonal linear subspaces $\V_1, \dots, \V_d$ of respective dimensions $\gamma_1, \dots, \gamma_d$~\cite{ye_optimization_2022,szwagier_rethinking_2023}. Let $\Pi_{\V_k}$ denote the orthogonal projection matrix onto $\V_k$. Then the eigendecomposition map $\lambda, \V \mapsto \sum_{k=1}^d \lambda_{q_k} \Pi_{\V_k}$ induces a diffeomorphism between $\K_+(\gamma) \times \Fl(\gamma)$ and $\Sym_+(\gamma)$. The submanifold and dimension properties then follow from those of $\K_+(\gamma)$ and $\Fl(\gamma)$.
\end{remark}

\noindent The previous result tells us that $\Sym_+(p)$ can be \textit{stratified} into submanifolds $\Sym_+(\gamma)$ of dimension $d + (p^2 - \sum_{k=1}^d \gamma_k^2) / 2$, as $\Sym_+(p) = \bigsqcup_{\gamma\in\C(p)} \Sym_+(\gamma)$. Therefore, for any $\Sigma\in\Sym_+(p)$, one can define $\operatorname{dim}(\Sigma)$ as the dimension of the \textit{unique} stratum $\Sym_+(\gamma)$ it belongs to. 
This implies that the \textit{penalized} covariance estimation problem~\eqref{eq:problem} can be replaced with a set of \textit{unpenalized} optimization problems on smooth matrix manifolds~\cite{ye_optimization_2022}.

In the Gaussian case, problem~\eqref{eq:problem} even has an \textit{explicit} solution, referred to as \textit{principal subspace analysis}~\cite{szwagier_curse_2024}.\footnote{
Principal subspace analysis was originally introduced to address a \textit{curse of isotropy} in PCA, in which close sample eigenvalues yield unstable principal components.
}
On each stratum $\Sym_+(\gamma)$, the maximum likelihood estimate $\hat\Sigma(\gamma)$ consists in the eigenvalue decomposition of the sample covariance matrix followed by a $\gamma$-block-averaging of the eigenvalues. Therefore, problem~\eqref{eq:problem} is equivalent to 
\begin{equation}
    \argmin_{\gamma\in\C(p)} \, \lrp{-2 \ln\L\bigl(\hat\Sigma(\gamma)\bigr) + \alpha \, \lrp{d + \frac{p^2 - \sum_{k=1}^d \gamma_k^2}2}},
\end{equation}
which can be solved \textit{exactly} as $\C(p)$ is finite.
Unfortunately, $\#\C(p) = 2^{p-1}$, so that solving exactly the problem is computationally prohibitive, even in moderate dimensions.
And this is not to mention non-Gaussian distributions, which do not even have an explicit maximum likelihood estimate of the covariance.

\section{Relaxation via Eigengap Sparsity}\label{sec:method}
We want to relax problem~\eqref{eq:problem} in order to make it computationally tractable. The first step consists in rewriting the parsimony-inducing penalty $\operatorname{dim}$ as $\ell^0$-norms of spectral gaps.
\begin{theorem}\label{thm:escp_1}
    Let $\gamma \coloneqq (\gamma_1, \dots, \gamma_d) \in \C(p)$ and $\Sigma \in \Sym_+(\gamma)$ with eigenvalues $~{\lambda \coloneqq (\lambda_1, \dots, \lambda_p) \in \K_+(\gamma)}$. Let $\delta\colon\R^2\rightarrow \R_{\geq 0}$ s.t. $\delta(a, b) = 0$ iff $a=b$. One has
    \begin{equation}\label{eq:thm_l0}
        \operatorname{dim}(\Sym_+(\gamma)) = 1 + \norm{\delta(\lambda_s, \lambda_{s+1})_{1 \leq s < p}}_0 + \norm{\delta(\lambda_s, \lambda_t)_{1 \leq s < t \leq p}}_0.
    \end{equation}
\end{theorem}
\begin{proof}
    By definition of $\lambda\in\K_+(\gamma)$, the eigenvalues can be divided into $d$ blocks of equal eigenvalues, of respective sizes $\gamma_1, \dots, \gamma_d$. 
    Let $q_k \coloneqq \sum_{l=1}^{k} \gamma_l$.
    Then one has $\delta(\lambda_s, \lambda_t) = 0 \iff \exists k\in [\![1, d]\!]\colon (s, t) \in  {]\!]q_{k-1}, q_k]\!]}^2$.
    Consequently, one has
    \begin{equation*}
        \norm{\delta(\lambda_s, \lambda_{s+1})_{1 \leq s < p}}_0 \coloneqq \#\{s\colon \delta(\lambda_s, \lambda_{s+1}) \neq 0\} = \#\{k\colon \lambda_{q_k} \neq \lambda_{q_k + 1}\} = d - 1.
    \end{equation*}
    Similarly, one has 
    \begin{equation*}
        \norm{\delta(\lambda_s, \lambda_t)_{1 \leq s < t \leq p}}_0 = \#\{ (s, t) \colon s<t, \, \lambda_s \neq \lambda_{t}\} = \sum_{k=1}^{d} \gamma_k \sum_{l=k+1}^d \gamma_{l}.
    \end{equation*}
    The last quantity can be rewritten by double summation inversion as
    \begin{equation*}
        \sum_{1\leq k < l \leq d} \gamma_k \gamma_l = \sum_{k=1}^{d} \gamma_k \lrp{p - \gamma_k - \sum_{l=1}^{k-1} \gamma_{l}} = p^2 - \sum_{k=1}^{d} \gamma_k^2 - \sum_{1\leq l < k \leq d} \gamma_k \gamma_l.
    \end{equation*}
    Hence, $\norm{\delta(\lambda_s,  \lambda_t)_{1\leq s < t \leq p}}_0 = (p^2 - \sum_{k=1}^{d} \gamma_k^2) / 2$, which concludes the proof. \qed
\end{proof}
Although relatively simple, Theorem~\ref{thm:escp_1} is quite remarkable.
While one may have thought of penalizing the successive gaps between eigenvalues (with total-variation-like penalties~\cite{tibshirani_sparsity_2005,harchaoui_multiple_2010}), our theorem conveys a greater message: to account for the parsimony of the \textit{eigenspaces} too, one should also penalize the gaps between non-adjacent eigenvalues. In other words, covariance parsimony tends to make \textit{all} the eigenvalues attracted to each other, which is reminiscent of the recent findings of Ledoit and Wolf in shrinkage estimation~\cite{ledoit_quadratic_2022}.
One can now naturally relax problem~\eqref{eq:problem} by replacing the $\ell^0$-norm with the $\ell^1$-norm.

Let $\K_+(p) = \{(\lambda_1, \dots, \lambda_p)\in\R^p\colon \lambda_1 \geq \dots \geq \lambda_p > 0\}$ denote the \textit{positive monotone cone}. Then the $\ell^1$-relaxation of problem~\eqref{eq:problem} is
    \begin{equation}\label{eq:thm_linear}
        \argmin_{\substack{V\in \O(p)\\ \lambda \in \K_+(p)}} \biggl(-2 \ln\L\lrp{V \diag{\lambda} V\T} + \alpha \, \sum_{s=1}^{p-1} \biggl(\delta(\lambda_s, \lambda_{s+1}) + \sum_{t=s+1}^p \delta(\lambda_s, \lambda_{t})\biggr)\biggr).
    \end{equation}
The relaxed problem~\eqref{eq:thm_linear} can be solved with classical constrained optimization algorithms, for instance alternating between the update of the eigenvectors in $V\in\O(p)$ and the eigenvalues in $\lambda\in\K_+(p)$. 
In the remainder of the paper, we focus on the fundamental case where $\L$ is a Gaussian likelihood.

\section{Projected Gradient Descent for Gaussian Densities}\label{sec:Gaussian}
Under the Gaussian assumption, the maximum likelihood estimate of the eigenvectors is explicit. The relaxed optimization problem~\eqref{eq:thm_linear} can then be simplified.
\begin{proposition}\label{prop:gaussian}
Let $X \coloneqq (x_1,\ldots,x_n)\in\R^{p\times n}$. Let $S \coloneqq \frac1n X X\T$ be the sample covariance matrix and $\ell_1 \geq \dots \geq \ell_p$ be its eigenvalues. Then the relaxed problem~\eqref{eq:thm_linear} is equivalent to
    \begin{equation}\label{eq:prop_gaussian}
        \argmin_{\lambda\in\K_+(p)} \, \lrp{\sum_{j=1}^p \lrp{\ln \lambda_j + \frac{\ell_j}{\lambda_j}} + \frac{\alpha}{n}\sum_{s=1}^{p-1} \lrp{\delta(\lambda_s, \lambda_{s+1}) + \sum_{t=s+1}^p \delta(\lambda_s, \lambda_{t})}}.
    \end{equation}
\end{proposition}
\begin{proof}
    Under the Gaussian assumption, the log-likelihood can be rewritten as $-2\ln\L\lrp{V \diag{\lambda} V\T} = n (p \ln(2\pi) + \sum_{j=1}^p \ln \lambda_j + \tr{\operatorname{diag}(\lambda)^{-1} V\T S V)}$~\cite{szwagier_curse_2024}.
    An optimal $\hat V \in \O(p)$ is any sequence of ordered eigenvectors of $S$ (with decreasing eigenvalues)~\cite{szwagier_curse_2024}. This implies that $\tr{\operatorname{diag}(\lambda)^{-1} {\hat V}\T S \hat V} = \sum_{j=1}^p \lambda_j^{-1} \ell_j$. We conclude the proof by dividing by $n$ and removing the constant $p\ln(2\pi)$. \qed
\end{proof}

\begin{remark}
    In the remainder of this paper, to ensure the compactness of the constraint set and the good conditioning of the covariance matrices~\cite{bickel_regularized_2008}, the cone $\K_+(p)$ is replaced with ${\K_\varepsilon(p) \coloneqq \{\lambda\in\R^p\colon 1 / \varepsilon \geq \lambda_1 \geq \dots \geq \lambda_p \geq \varepsilon\}}$ with $~{\varepsilon \coloneqq 10^{-10}}$. It turns out to have little effect in practice, although it sometimes does in high dimensions or when some population eigenvalues are small.
\end{remark}

\noindent We propose to solve problem~\eqref{eq:prop_gaussian} on $\K_\varepsilon(p)$ with a \textit{projected gradient descent} (see e.g.~\cite{OlikierWaldspurger} and references therein) summarized in Algorithm~\ref{alg:pgd}.
\begin{algorithm}[b]
\caption{Projected Gradient Descent for Gaussian Eigengap Sparsity}\label{alg:pgd}
\begin{algorithmic}
\Require $\ell \in \K_0(p), Q\in\O(p), \alpha\in\R_{\geq 0}$ \Comment{sample eigenvalues/vectors, hyperparameter}
\State $\lambda \gets \ell$ \Comment{initialization with sample eigenvalues}
\For{$i$ = 1, 2, \dots}
    \State $g \gets \frac{\lambda - {\ell}} {\lambda^2} + \frac{\alpha}{n}\partial_\lambda \sum_{s=1}^{p-1} \lrp{\delta(\lambda_s, \lambda_{s+1}) + \sum_{t=s+1}^p \delta(\lambda_s, \lambda_{t})}$
    \Comment{gradient (elementwise)}
    \State $\lambda \gets \Pi_{\K_\varepsilon(p)}(\lambda - \beta \, g)$ \Comment{backtracking projected line search~\cite{OlikierWaldspurger}}
\EndFor
\Ensure $\hat \Sigma \coloneqq  Q \diag{\lambda} {Q}\T$ \Comment{optimal covariance matrix}
\end{algorithmic}
\end{algorithm}
A central observation is that the orthogonal projection onto the constraint set $\K_\varepsilon(p)$ corresponds to the well-known problem of \textit{isotonic regression} (cf.~\cite{salmon_isotonic_2024} for a pedagogical explanation and illustration).
This quadratic program can be \textit{exactly} solved in \textit{linear time} via the Pools Adjacent Violators Algorithm (PAVA)~\cite{best_active_1990}, which roughly boils down to sequentially investigating the eigenvalues and block-averaging the ones that violate the ordering constraints. It is followed by an eigenvalue thresholding to ensure that the eigenvalues fit the bounds in $\K_\varepsilon(p)$.

Consequently, at each iteration, Algorithm~\ref{alg:pgd} somewhat makes a tradeoff between increasing the Gaussian likelihood and decreasing the eigengaps via a gradient step, and subsequently equalizes the eigenvalues that ``intersect'' (i.e. that cross the boundary of the monotone cone) via an isotonic regression.
While isotonization has been historically employed by Stein as a post-shrinkage trick to avoid numerical issues caused by close eigenvalues (see e.g.~\cite{ledoit_quadratic_2022}), it naturally arises in our Algorithm~\ref{alg:pgd} as a projected gradient step. 
Therefore, while close eigenvalues are a \textit{curse} in the classical shrinkage literature, they are a \textit{blessing} in our framework since they get equalized and consequently bring parsimony.

The choice of eigengap function $\delta$ has a great impact on the results. While the most natural choice seems to be the \textit{absolute eigengap} $\delta(a, b) \coloneqq |a - b|$, we conjecture that it tends to make \textit{all} the eigenvalues equal (in the Gaussian case), which is not always desired.
Consequently, we choose to consider the \textit{relative eigengap} $\delta(a, b) \coloneqq (a - b) / a$ (for $a \geq b$), which is strongly justified by theory~\cite{szwagier_curse_2024} and which also turns out to yield better experimental results.

\section{Experiments}\label{sec:experiments}
We run some synthetic experiments (with 10 repetitions) to compare our Eigengap Sparsity for Covariance Parsimony (ESCP, with $\alpha_n = \ln n$) to the original Principal Subspace Analysis (PSA)~\cite{szwagier_curse_2024}, which solves exactly problem~\eqref{eq:problem} for a Gaussian density (cf. section~\ref{sec:geometry}).
We also include for completeness the Sample Covariance Matrix (SCM) and the Ledoit--Wolf (LW) estimators~\cite{ledoit_well-conditioned_2004}, although they do not optimize the same objective function. Indeed, SCM optimizes the unpenalized-likelihood problem~\eqref{eq:problem} (with $\alpha=0$) and LW looks for a linear combination of the sample covariance and identity matrices that minimizes the expected Frobenius distance to the population covariance matrix.\footnote{The code is written in Python 3.9 and run on an \textit{Intel® Core™ i7-10850H} CPU with 16 GB of RAM: \url{https://github.com/tomszwagier/eigengap-sparsity}.}

We sample $n$ points from a Gaussian distribution of covariance $\Sigma\in\Sym_+(p)$ under three settings: 
(a) $n=40, p=20, \Sigma=I_{20}$, 
(b) $n=200, p=100, \Sigma=I_{100}$ and 
(c) $n=400, p=200, \Sigma=\operatorname{diag}(10 \, I_{80}, 1 \, I_{80}, 0.1 \, I_{40})$, the last covariance matrix being denoted by $\Xi_{80, 80, 40}$ for short. 
The last setting is inspired from the baseline scenario of Ledoit and Wolf~\cite{ledoit_quadratic_2022}, which is claimed to be an interesting and difficult case. 
The first two settings have isotropic covariance matrices, therefore dimension $1$, while the last setting has dimension $3 + (200^2 - 80^2 - 80^2 - 40^2) / 2 = 12803$.
For each method, we report the penalized-likelihood $L_{\mathrm{p}}(\hat\Sigma) \coloneqq -2 \ln\L(\hat\Sigma) + \alpha_n \, \operatorname{dim}(\hat\Sigma)$, the covariance estimation error $L_\mathrm{F}(\hat\Sigma, \Sigma) \coloneqq \|{\hat\Sigma - \Sigma}\|_\mathrm{F}^2 / p$ and the number of parameters $\operatorname{dim}(\hat \Sigma)$ (cf. Table~\ref{tab:results}). We also report in Fig.~\ref{fig:results} the running times and the estimated scree plots.
\begin{table}[t]
\caption{Evaluation of the penalized-likelihood $L_{\mathrm{p}}(\hat\Sigma)$, the covariance estimation error $L_\mathrm{F}(\hat\Sigma, \Sigma)$ and the number of parameters $\operatorname{dim}(\hat \Sigma)$ for different covariance estimation methods: Sample Covariance Matrix (SCM), Ledoit--Wolf (LW)~\cite{ledoit_well-conditioned_2004}, Principal Subspace Analysis (PSA)~\cite{szwagier_curse_2024} and Eigengap Sparsity for Covariance Parsimony (ESCP) on multivariate Gaussian datasets with varying $(n, p, \Sigma)$.}\label{tab:results}
\begin{subtable}{.32\linewidth}\centering
\caption{$(40, 20, I_{20})$}
{\begin{tabular}{|c|c|c|c|}
\hline
Method &  $L_{\mathrm{p}}$ & $L_\mathrm{F}$ & $\operatorname{dim}$\\
\hline
SCM             & 33.2 & 0.544 & 210\\ 
LW              & 31.3 & 0.008 & 126\\
PSA             & 20.4 & 0.008 & 3\\
ESCP            & 20.5 & 0.004 & 3\\
\hline
\end{tabular}}
\end{subtable}
\begin{subtable}{.32\linewidth}\centering
\caption{$(200, 100, I_{100})$}
{\begin{tabular}{|c|c|c|c|}
\hline
Method &  $L_{\mathrm{p}}$ & $L_\mathrm{F}$ & $\operatorname{dim}$\\
\hline
SCM             & 202 & 0.50807 & 5050\\
LW              & 193 & 0.00030 & 3535\\
PSA             & -- & -- & --\\
ESCP            & 100 & 0.00004 & 1\\
\hline
\end{tabular}}
\end{subtable}
\begin{subtable}{.32\linewidth}\centering
\caption{$(400, 200, \Xi_{80, 80, 40})$}
{\begin{tabular}{|c|c|c|c|}
\hline
Method &  $L_{\mathrm{p}}$ & $L_\mathrm{F}$ & $\operatorname{dim}$\\
\hline
SCM             & 532 & 9.9 & 20100\\
LW              & 661 & 6.7 & 20100\\
PSA             & -- & -- & --\\
ESCP            & 455 & 2.3 & 13508\\
\hline
\end{tabular}}
\end{subtable}
\end{table}
\begin{figure}[b]
\centering
\includegraphics[width=\textwidth]{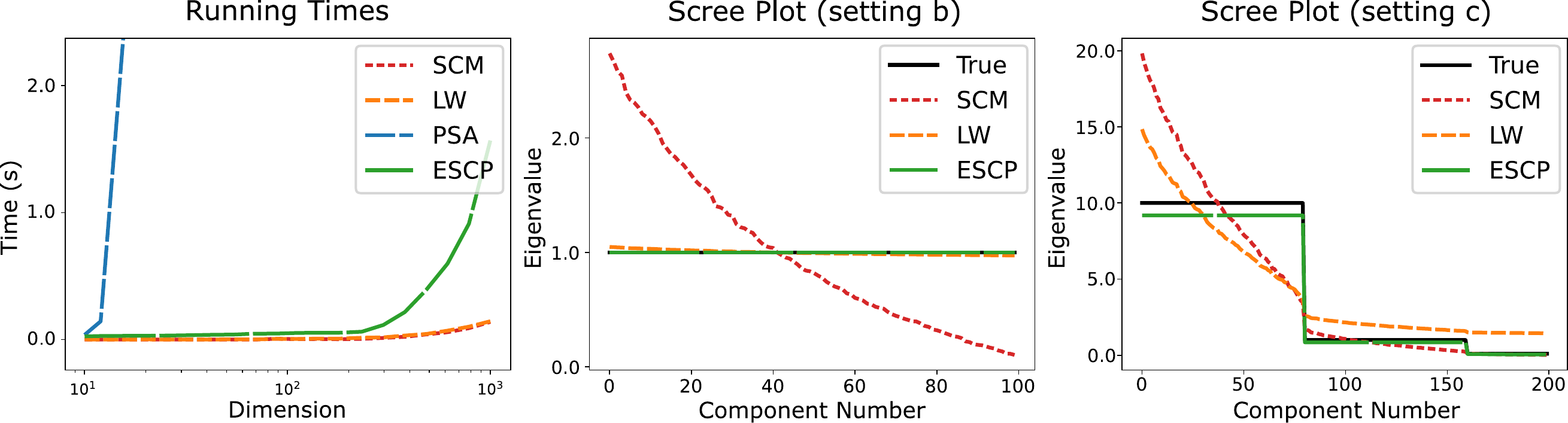}
\caption{Evaluation of the covariance estimation methods for varying $(n, p, \Sigma)$.
(Left) Running times in seconds, for increasing dimension $p$ in logarithmic scale.
(Middle, Right) Scree plots of estimated eigenvalues in the settings (b) and (c).
}
\label{fig:results}
\end{figure}

There are many observations to make. First, SCM performs poorly in terms of covariance estimation ($L_\mathrm{F}$), even for $p = 20$. This can be explained by the eigenvalue-spreading phenomenon~\cite{ledoit_well-conditioned_2004}, which we indeed observe in Fig.~\ref{fig:results} (Middle)---where the sample eigenvalues are highly spread around their true value, which is 1. 
LW is much better in terms of covariance estimation ($L_\mathrm{F}$, which is what it is designed for), although it may have a higher penalized-likelihood ($L_{\mathrm{p}}$) since it shrinks the eigenvalues (so it reduces the likelihood with respect to SCM which is the maximum likelihood estimator) while not necessarily reducing the number of parameters. 
As expected, PSA yields the best penalized-likelihood ($L_{\mathrm{p}}$) since it finds its exact minimum. 
However, its running time is extremely large (cf. Fig.~\ref{fig:results}, Left), which makes it computationally prohibitive for $p \gtrsim 20$. 
In contrast, ESCP yields similar performance as PSA in \textit{drastically} lower time, which makes it possible to run in higher dimensions.
In dimensions 100 and 200 (Table~\ref{tab:results} (b, c) and Fig.~\ref{fig:results} (Middle, Right)) ESCP has a lower penalized-likelihood ($L_{\mathrm{p}}$) than both SCM and LW, which is not much surprising since it optimizes a relaxation of $L_{\mathrm{p}}$. The good surprise is on the covariance estimation results: ESCP achieves a better estimation in average than LW although it is not designed to minimize $L_\mathrm{F}$. This can be explained by looking more closely at Fig.~\ref{fig:results}. One can indeed see that ESCP almost perfectly recovers the true covariance matrix $I_{100}$ (b) and recovers well the piecewise-constant-eigenvalue matrix $\Xi_{80, 80, 40}$ (c), which we believe to be a challenging problem that cannot be properly addressed by shrinkage methods yet.\footnote{
As we can see, the average number of parameters ($13508$) is higher than the true one ($12803$), meaning that a few sample eigenvalues were not properly equalized.
}
Putting this into perspective with the {curse of isotropy} phenomenon~\cite{szwagier_curse_2024}, we expect the underlying principal subspaces to be much more stable and interpretable than the sample eigenvectors output by shrinkage methods.
Hence, investigating covariance estimation through the prism of covariance parsimony (induced by the equalization of close eigenvalues) seems like a promising perspective.

\section{Conclusion}\label{sec:conclusion}
We formulated the problem of parsimonious covariance estimation as a penalized-likelihood optimization with an eigengap-sparsity-inducing penalty. The latter was derived from the $\ell^1$-relaxation of intrinsic dimensionality in the stratification of symmetric positive-definite matrices by the multiplicities of the eigenvalues.
We proposed a projected gradient descent algorithm for the Gaussian case, which boils down to iteratively updating the ordered-eigenvalues via a penalized-likelihood gradient-step and equalizing the intersecting eigenvalues via an isotonic regression which automatically induces parsimony. We illustrated the good performance of our method in terms of speed (with respect to classical model selection~\cite{szwagier_curse_2024}), parsimony and covariance estimation (with respect to the sample covariance matrix and shrinkage methods~\cite{ledoit_well-conditioned_2004}).

The eigengap sparsity penalty is only at an early stage of research and many perspectives arise: evaluating the performance of our method on more diverse settings (eigenvalue profiles, non-Gaussian distributions) and choices of hyperparameters (regularization strength $\alpha$ and eigengap function $\delta$); making the learning algorithm more robust and efficient in high dimensions using flag manifold optimization~\cite{ye_optimization_2022,szwagier_nested_2025}; studying the theoretical guarantees of the algorithm; considering Riemannian metrics on $\Sym_+(p)$ or $\K_+(p)$ beyond the embedded one etc. 
One could also extend our penalized-likelihood estimation framework beyond covariance matrices, considering for instance graph Laplacians and Hessians, whose parsimony is very important~\cite{lombaert_focusr_2013,lambert_limited-memory_2023}.

Our proposed approach draws some bridges with active areas of research. 
First, the eigengap sparsity can be of interest to the sparse optimization community~\cite{bach_optimization_2012}. 
Our algorithm moreover surprisingly induces sparsity via an isotonic regression instead of classical thresholding operators like in LASSO~\cite{tibshirani_regression_1996}. 
Second, the eigengap penalty results from a relaxation of covariance parsimony, which addresses the eigenvalue-spreading phenomenon studied in pioneering shrinkage methods~\cite{ledoit_well-conditioned_2004}. Interestingly, while most shrinkage methods focus on the eigenvalues, our framework for covariance estimation includes both the eigenvalues and the flag of eigenspaces. Moreover, while low-rank models like PPCA~\cite{tipping_probabilistic_1999} only group the smallest eigenvalues (considered as noise), our models may also group some of the largest eigenvalues, therefore better taking into account the eigenvalue-spreading phenomenon. 
Third, our work has interesting links with the \textit{elasso} method from David E. Tyler and Mengxi Yi~\cite{tyler_lassoing_2020} and follow-up works~\cite{basiri_fusing_2019}. Although the motivations, the penalties and the optimization algorithms are quite different from ours, these methods also result in an automatic grouping of the eigenvalues. We believe that they could be seamlessly integrated in our framework by modifying the eigengap function $\delta$ and taking advantage of automatic differentiation.

\begin{credits}
\subsubsection{\ackname} 
This work was supported by the ERC grant \#786854 G-Statistics from the European Research Council under the European Union’s Horizon 2020 research and innovation program and by the French government through the 3IA Côte d’Azur Investments ANR-23-IACL-0001 managed by the National Research Agency.

\subsubsection{\discintname} 
The authors have no competing interests to declare that are relevant to the content of this article.
\end{credits}
%
%
%
\bibliographystyle{splncs04}
\bibliography{sample}

\end{document}